\documentclass[conference]{IEEEtran}
\IEEEoverridecommandlockouts
\usepackage{cite}
\usepackage{amsmath}
\usepackage{amsfonts}
\usepackage{amssymb}
\usepackage{amsthm}
\usepackage{tikz}
\usepackage{cases}
\usepackage{graphicx}
    \graphicspath{{../}}
    \DeclareGraphicsExtensions{.pdf}
\usepackage[caption=true,font=footnotesize]{subfig}
\usepackage{multirow}
\usepackage{makecell}
\usepackage{booktabs}
\usepackage{url}
\usepackage{xtab}
\usepackage{array}
\usepackage{algorithm}
\usepackage{algorithmic}
\usepackage{enumerate}
\usetikzlibrary{arrows}
\usepackage{caption}
\theoremstyle{plain}

\newtheorem{theorem}{Theorem}
\newtheorem{lemma}[theorem]{Lemma}

\newtheorem{definition}[theorem]{Definition}

\theoremstyle{definition}
\newtheorem{example}{Example}

\newtheorem{remark}{Remark}

\begin{document}
\title{A Capacity-Achieving $T$-PIR Scheme Based On MDS  Array Codes}

\author{\IEEEauthorblockN{Jingke Xu,~ Yaqian Zhang,~ Zhifang Zhang}\\
\IEEEauthorblockA{\fontsize{9.8}{12}\selectfont KLMM, Academy of Mathematics and Systems Science, Chinese Academy of Sciences, Beijing 100190, China\\
 School of Mathematical Sciences, University of Chinese Academy of Sciences, Beijing 100049, China\\
Emails: xujingke14@mails.ucas.edu.cn, zhangyaqian15@mails.ucas.ac.cn, zfz@amss.ac.cn}
}
\maketitle
\thispagestyle{empty}

\begin{abstract}
Suppose a database containing $M$ records is replicated in each of $N$ servers, and a user wants to privately retrieve one record by accessing the servers such that identity of the retrieved record is secret against any up to $T$ servers. A scheme designed for this purpose is called a $T$-private information retrieval ($T$-PIR) scheme.

In this paper we focus on the field size of  $T$-PIR schemes.
We design a general capacity-achieving $T$-PIR scheme whose queries are generated by using some {\rm MDS } array codes. It only requires
field size $q\geq\sqrt[\ell]{N}$, where $\ell=\min\{t^{M-2},(n-t)^{M-2}\}, ~t=T/{\rm gcd}(N,T),~n=N/{\rm gcd}(N,T)$ and has the optimal sub-packetization $Nn^{M-2}$. Comparing with  existing capacity-achieving $T$-PIR schemes, our scheme has the following advantage, that is, its field size monotonically decreases as the number of records $M$ grows. In particular,  the binary field is sufficient for building a capacity-achieving T-PIR scheme as long as $M\geq 2+\lceil\log_\mu\log_2N\rceil$, where $\mu=\min\{t,n-t\}>1$.
\end{abstract}

\section{Introduction}\label{sec1}
Private information retrieval (PIR) is a canonical problem in the study of privacy issues that arise from the retrieval of
information from public databases. Typically, a PIR model involves a database containing $M$ records
stored across $N$ servers and a user who wants to privately retrieve one record by accessing the servers. Specifically, the privacy
requirement means any colluding subset containing no more than $T$ servers knows nothing about identity of the retrieved record.
Since it is closely related to cryptography \cite{BIKO12ACCC:SharePIR} and coding theory \cite{Yekhanin07PHD:LDC&PIR}, PIR has
become a central research topic in the computer science literature since it was first introduced by Chor et al.\cite{CKGS95FOCS:PIR}
in 1995.

The efficiency of PIR scheme is characterized by its {\it rate}.
Specifically, the {\it rate} of a PIR scheme is measured as the ratio between the
retrieved data size and the downloaded size, and the {\it capacity} is defined as the supremum of the rate over all PIR schemes. Recently, much work has been done on determining the capacity of PIR in various cases.
Sun and Jafar derived that the capacity for the non-colluding servers (i.e., $T=1$) is  {\small$\frac{1-1/N}{1-(1/N)^{M}}$} in \cite{Sun&Jafar16:CapacityPIR} and further proved that the capacity for the colluding servers (i.e., $T>1$) is {\small$\frac{1-T/N}{1-(T/N)^{M}}$} in \cite{Sun&Jafar16:ColludPIR}. Moreover, the latter is called $T$-PIR.  They also determined the capacity of PIR with symmetric privacy
in \cite{Sun&Jafar16:CapaSymmPIR}.
The capacity of PIR  with MDS coded non-colluding servers is determined  in
\cite{Bana&Uluk16:CapacityPIRCoded}.  It remains an open problem to determine the capacity of PIR with MDS coded colluding servers. For non-MDS coded storage,  PIR schemes with colluding or non-colluding servers are presented
in \cite{Lin&Kumar&R&Amat18:CapaPIRNonMDS,HGHTKK18:TPIRwithTranCodes}.

In general, the capacity of PIR is achieved by dividing each record into multiple sub-packets and querying from each server specially designed combinations of these sub-packets. Therefore, both the number of sub-packets and the size of each sub-packet are important metrics for measuring the implementation complexity of a PIR scheme. As to the former, we call the number of sub-packets contained in each record as sub-packetization. The optimal sub-packetization for capacity-achieving PIR schemes has been determined in some cases \cite{Sun&Jafar16:OptimalPIR,Zhang&Xu17:OptimalSubpacketization,Xu&Zhang18:OptimalSubpacketization}. As to the latter, since all existing PIR schemes are linear schemes over some finite fields, it is actually about the size of the field on which the PIR scheme can be built. The main concern of this work is to reduce the field size for $T$-PIR schemes while maintaining the rate achieving the capacity and the optimal sub-packetization.

In \cite{Sun&Jafar16:ColludPIR}, it requires a field of size $q=\Omega(N^{2}T^{M-2})$ for the capacity-achieving $T$-PIR scheme. The field size is reduced to $q=\Omega(Nt^{M-2})$ for the capacity-achieving $T$-PIR scheme with optimal sub-packetization in \cite{Zhang&Xu17:OptimalSubpacketization}, where $t=T/{\rm gcd}(N,T)$. The best known result of field size for capacity-achieving $T$-PIR scheme is $q=\Omega(N)$ in \cite{Xu&Zhang:PIRSmallfield}.
But the field size is still unfriendly  with the growth of the number of servers.

The main contribution of this work consists of designing a $T$-PIR scheme that simultaneously achieves the capacity and the optimal sub-packetization $\!Nn^{M\!-\!2}\!$ over a finite field $\mathbb{F}_q$ for all possible parameters $(N,T,M)$, and it requires the field size $q \geq\sqrt[\ell]{N}$,  where $\ell=\min\{t^{M-2},(n-t)^{M-2}\}, n=N/{{\rm gcd}(N,T)},t=T/{{\rm gcd}(N,T)}.$  When $\ell=1$, the constraint of the field size in our scheme degenerates into $q\geq N$, which is the same with  that of the capacity-achieving $T$-PIR scheme in \cite{Xu&Zhang:PIRSmallfield}. When $\ell>1$, the binary field is sufficient for building a capacity-achieving $T$-PIR scheme provided  $M\geq 2+\lceil\log_\mu\log_2N\rceil$, where $\mu=\min\{t,n-t\}$.

\begin{table}[ht]
\centering
\setlength{\abovecaptionskip}{0.2cm}
\setlength{\belowcaptionskip}{-0.1cm}
\begin{tabular}{|r|l|}
  \hline
  Reference& Field size($q$) \\\hline
   Sun et al. \cite{Sun&Jafar16:ColludPIR}&$q\geq\max\{N^2T^{M-2},N^2(N-T)^{M-2}\}$\\\hline
   Zhang et al. \cite{YZhangGGe17:CodedTPIR}&$q\geq\max\{Nnt^{M-2},Nn(n-t)^{M-2}\}$\\\hline
   Zhang et al. \cite{Zhang&Xu17:OptimalSubpacketization}&$q\geq\max\{Nt^{M-2},N(n-t)^{M-2}\}$\\\hline
  Xu et al. \cite{Xu&Zhang:PIRSmallfield} &$q\geq N$\\\hline
  This paper&$q\geq \sqrt[\ell]{N}, \ell=\min\{t^{M-2},(n-t)^{M-2}\}$\\
  \hline
\end{tabular}
\caption{A list of all existing  capacity-achieving $T$-PIR schemes with $T\geq2$.  And  $n=\frac{N}{{\rm gcd}(N,T)},t=\frac{T}{{\rm gcd}(N,T)}.$ }\label{tab0}
\end{table}

Comparing with  all existing capacity-achieving $T$-PIR schemes with $T\geq 2$ in \cite{Sun&Jafar16:ColludPIR},
\cite{YZhangGGe17:CodedTPIR},\cite{Zhang&Xu17:OptimalSubpacketization}, \cite{Xu&Zhang:PIRSmallfield}, as displayed in Table \ref{tab0},  the main difference in our scheme is to employ {\rm MDS} array codes to generate queries, which is a key idea for reducing the field size.  Moreover, an advantage of our scheme is  that  its field size monotonically decreases as the number of records $M$ grows.

The rest of this paper is organized as follows. First, the $T$-PIR model is formally introduced  and the  {\rm MDS} array code is defined in
Section \ref{sec2}. Then in Section III  an example of the $T$-PIR scheme is presented  to explain the design idea. The recovery property of {\rm MDS} array codes is proved and  the general descriptions of our scheme are given in Section IV. Finally, Section V concludes the paper.

\section{Preliminaries}\label{sec2}
\subsection{Notations and the $T$-PIR model}
For an integer $n\!\in\!\mathbb{N}$, we denote by $[n]$ the set $\{1,...,n\}$. For a vector ${\bf u}=(u_1,...,u_n)$ and a subset $\Gamma=\{i_1,...,i_m\}\subseteq [n]$, denote ${\bf u}_\Gamma=(u_{i_1},...,u_{i_m})$. Most vectors in this paper are row vectors and they are denoted by the bold lowercase letters (eg. ${\bf a,b}$). For a block matrix $A=(A^{(1)},A^{(2)},...,A^{(N)})$ and $\Gamma=\{i_1,...,i_m\}\subseteq [N]$, denote $A^\Gamma=(A^{(i_1)},...,A^{(i_m)})$.

Suppose there are $M$ records $W_1,...,W_M$ and $N$ servers ${\rm Serv}^{(1)},...,{\rm Serv}^{(N)}$, each server stores all the $M$ records. Moreover, the records are independent and each can be seen as an $L$-length vector over $\mathbb{F}_q$.  Then suppose a user wants to privately retrieve $W_\theta$ for some $\theta\!\in\![M]$. Formally, a $T$-PIR scheme consists of two phases:
\begin{itemize}
  \item {\bf Query phase.} Given $\theta\in[M]$,  the user generates the query {\small${\rm Que}(\theta,S)=(Q_\theta^{(1)},...,Q_\theta^{(N)})$}, and sends {\small$Q_\theta^{(j)}$} to {\small${\rm Serv}^{(j)}$} for $1\!\leq\! j\!\leq\!N$, where $S$ are some random resources privately chosen by the user. Note that {\small${\rm Que}(\cdot,\cdot)$} is the {\it query function} defined by the scheme.
  \item {\bf Response phase. }After receiving $Q_\theta^{(j)}$,  the ${\rm Serv}^{(j)}$ computes the answers ${\rm Ans}^{(j)}(Q_\theta^{(j)},W_{[M]})=A_\theta^{(j)}$ for $1\!\leq\!j\!\leq\!N$, and sends it back to the user, where ${\rm Ans}^{(j)}(\cdot,\cdot)$ is the {\it answer  function} defined by the scheme.
\end{itemize}
 Moreover the functions ${\rm Que}(\cdot,\cdot)$ and ${\rm Ans}^{(j)}(\cdot,\cdot), 1\!\leq\!j\!\leq\!N$  must  satisfy the following two conditions:
\begin{itemize}
\item[(1)]{\it Correctness: }The user can reconstruct $W_\theta$ after collecting all answers from the $N$ servers, i.e., {\small$H(W_\theta|A^{[N]}_\theta\!,Q^{[N]}_\theta\!,S,\theta)\!=\!0$}, where {\small$H(\cdot|\cdot)$} is the conditional entropy.
\item[(2)]{\it Privacy: } For any $\Gamma\!\subseteq\![N]$ with $|\Gamma|\!=\!T$, the serves in $\Gamma$ can't obtain any information on $\theta$ even if they collude with each other,  i.e., {\small$I(\theta;Q^{\Gamma}_\theta,A^{\Gamma}_\theta,W_{[M]})=0$}, where {\small$I(\cdot~;\cdot)$} denotes the mutual information.

\end{itemize}

Define the {\it rate} $\mathcal{R}$ of a $T$-PIR scheme by
{\small$$\mathcal{R}=\frac{H(W_\theta)}{\sum_{i=1}^{N}H(A^{(i)}_\theta)}\;,$$}that is, $\mathcal{R}$ characterizes the amount of retrieved information per unit of downloaded data. Furthermore, the {\it capacity} of $T$-PIR is defined by the largest  rate over all achievable $T$-PIR schemes, denoted by $\mathcal{C}_{\mbox{\tiny $T$-PIR}}$.
By \cite{Sun&Jafar16:ColludPIR}, it has that {\small$\mathcal{C}_{\mbox{\tiny $T$-PIR}}=\frac{1-T/N}{1-(T/N)^{M}}$}.

\subsection{MDS Array Codes}
In this section we introduce {\rm MDS}  array code  used in this paper and then give a method to construct such code over $\mathbb{F}_q$.

Suppose $N>T\geq 1$ and $N,T$ are two positive integers. For a linear $[N,T]$ code $\mathcal{C}$ over $\mathbb{F}_{q^\ell}$, a codeword ${\bf c}=(c_1,c_2,...,c_N)$ can be seen as an $N\ell$-length vector ${\bf c}=({\bf c}_1,{\bf c}_2,...,{\bf c}_N)$ over $\mathbb{F}_q$, i.e., for $i\in[N]$, the code block ${\bf c}_i=(c_{i,1},c_{i,2},...,c_{i,\ell})\in\mathbb{F}^{\ell}_{q}$ denotes the $\ell$-length vector corresponding to the symbol $c_i\in\mathbb{F}_{q^\ell}$. So we call the code $\mathcal{C}$ a linear array code over $\mathbb{F}_{q}$, and refer to the code as an $(N,T;\ell)_{q}$ linear array code. Equivalently, an $(N,T;\ell)_{q}$ linear array code can be defined by a $T\ell \times N\ell$ full rank matrix $G$ over $\mathbb{F}_q$ as follows,
$$
\mathcal{C}=\{{\bf c}=({\bf m}_1,{\bf m}_2,...,{\bf m}_T)G: ({\bf m}_1,{\bf m}_2,...,{\bf m}_T)\in \mathbb{F}^{T\ell}_{q}\}.
$$
The matrix $G$ is called a generator matrix of the array code $\mathcal{C}$. Then the  generator matrix $G$ can be viewed as a block matrix $$G=(G^{(1)},G^{(2)},...,G^{(N)}).$$
For $i\in[N]$,  the $T\ell\times \ell$ sub-matrix $G^{(i)}$ is represented as the thick column associated with the $i$th code block in the codewords of $\mathcal{C}$.
\begin{definition}\label{def2}(MDS Array Codes) A linear array code $\mathcal{C}$ over $\mathbb{F}_q$ is called an $(N,T;\ell)$ MDS array code if its generator matrix ${\small G=(G^{(1)},G^{(2)},...,G^{(N)})\in\mathbb{F}_q^{T \ell\times N\ell}}$ has the following {\rm MDS} property:
\begin{equation}\label{eqmds}
\forall~\Gamma \subseteq [N] ~\text{with}~ |\Gamma|=T, {\rm rank}(G^{\Gamma})=T\ell
\end{equation}
where $G^{(i)}\in\mathbb{F}_q^{T \ell\times \ell}$ for $i\in[N]$ and $T<N$. \end{definition}
By the definition of $(N,T;\ell)$ {\rm MDS} array code $\mathcal{C}$, it degenerates into a {\rm MDS} code over $\mathbb{F}_q$ for $\ell=1$. Next we give a method to construct an $(N,T;\ell)$ MDS array code.

Suppose $\alpha$ is a primitive element of $\mathbb{F}_{q^\ell}$, then $\mathbb{F}_{q^\ell}=\{\alpha^j:0\leq j \leq q^\ell-2\}\cup \{0\}$. Suppose $m(x)$ is the minimal polynomial of $\alpha$ over $\mathbb{F}_q$. Let $C\in \mathbb{F}_q^{\ell\times\ell}$ be the companion matrix of $m(x)$ and $\mathbb{F}=\{C^j:j\in\mathbb{Z}\}\cup\{\bf 0\}$. Then $\mathbb{F}$ is a finite field of size $q^\ell$  and  the map which is defined by  $\varphi(\alpha^j)=C^j$ and $\varphi(0)={\bf 0}$ is a field isomorphism from $\mathbb{F}_{q^\ell}$ to $\mathbb{F}$  by \cite{Finitefield}.
Let {\small$G=(\alpha_{i,j})\in \mathbb{F}_{q^{\ell}}^{T\times N}$} be a generator matrix of an $[N,T]$ MDS code over
 $\mathbb{F}_{q^\ell}$. Note that each symbol of $\mathbb{F}_{q^\ell}$ can be represented as an
 $\ell\times\ell$ matrix in $\mathbb{F}$ over $\mathbb{F}_q$ by using the field isomorphism $\varphi(\cdot)$, then the matrix $G$ can be seen as an $T\times N$ block matrix $G'$, i.e., $G^\prime=(\varphi(\alpha_{i,j}))_{i\in[T],j\in[N]}$, and each thick column is an $T\ell \times \ell$ matrix over
 $\mathbb{F}_q$. It is easy to verify that for any $\Gamma\subseteq[N]$ with $|\Gamma|=T$,
 $$\det ((\varphi(\alpha_{i,j}))_{i\in[T],j\in\Gamma})=\varphi(\det((\alpha_{i,j})_{i\in[T],j\in\Gamma}))\neq 0.$$
 Hence the linear array code which is defined by the generator matrix $G'$ over
 $\mathbb{F}_q$ is an $(N,T;\ell)$ MDS array code. Then we can directly obtain the following theorem.
 \begin{theorem}\label{thm0} Suppose $q$ is a power of a prime and $N,T,\ell \in \mathbb{N}$ with $N>T\geq 2$.
 If $q^\ell\geq N$, then there exists an $(N,T;\ell)$ {\rm MDS} array code over $\mathbb{F}_q$.
 \end{theorem}

Recall that for  existing capacity-achieving $T$-PIR schemes in \cite{Sun&Jafar16:ColludPIR},
\cite{YZhangGGe17:CodedTPIR},\cite{Zhang&Xu17:OptimalSubpacketization}, some $[N\ell_k,T\ell_k]$ {\rm MDS} codes over $\mathbb{F}_q$ are used to construct the query. And all $N\ell_k$ symbols of each codeword are equally divided into $N$ blocks. Then this {\rm MDS} code can be seen as an $(N,T;\ell_k)$ MDS array code over $\mathbb{F}_q$. Based on this observation, we find a direction to reduce the field size. That is, we generate the query  by using some {\rm MDS} array codes over a smaller finite field rather than {\rm MDS} codes. Moreover, the {\rm MDS} array codes need to satisfy some special property that is determined by the correctness condition. To formally illustrate this idea, we will give an example in the next section.

\section{Example For $N=5,T=3, M=3$}\label{sec3}
Before constructing our schemes, we first give an example by using the method described in \cite{Zhang&Xu17:OptimalSubpacketization}.  And then, we explain how to reduce the field size by modifying this scheme.

\begin{example}\label{eg1} Suppose $M=3$, $N=5$ and $T=3$. The field size $q \geq 15$ is enough and the sub-packetization of this case is $L=N^{M-1}=25$, so each record can be seen as a $25$-dimensional vector over $\mathbb{F}_{q}$, i.e., $W_1,W_2,W_3\in\mathbb{F}_{q}^{25}$. WLOG, suppose $W_1$ is the desired record, i.e., $\theta=1$.

Let $S_1,S_2,S_3$ be three matrices chosen by the user independently and uniformly from all $25\times25$ invertible matrices over $\mathbb{F}_q$. Actually, $S_1,S_2$ and $S_3$ are the random resources privately held by the user. Then, define
   \begin{equation}\label{eqb}
   \begin{split}
   &(a_1,a_2,...,a_{25})=W_1S_1 \\
   &(b_1,b_2,...,b_{25})=W_2S_2[:,(1:15)]G \\
   &(c_1,c_2,...,c_{25})=W_3S_3[:,(1:15)]G
  \end{split}
\end{equation}
 where $S_i[:,(1:15)]$ denotes the $25\times 15$ matrix formed by the first $15$ columns of $S_i$ and  {\small$G={\footnotesize\begin{pmatrix}
G_1&0\\0&G_2
\end{pmatrix}}$.} Moreover, {\small$G_1\in\mathbb{F}_q^{9\times 15}$} is a generator matrix of an {\small$[15,9]$ {\rm MDS}} code and {\small$G_2\in\mathbb{F}_q^{6\times 10}$} is  a generator matrix of an {\small$[10,6]$ {\rm MDS}} code over $\mathbb{F}_q$.

 It can be seen that the answers are all sums of the symbols $a_{i},b_{i},c_{i}$ in Fig.\ref{fg1}. For each sum $x$ in Fig.\ref{fg1}, we define its support as a subset of $[3]$ and this subset is composed of the label of all terms in the sum $x$, denoted by ${\rm supp}(x)$. For example, ${\rm supp}(a_i)=\{1\},{\rm supp}(b_i+c_j)=\{2,3\}$. For any $\Lambda\subseteq[3]$, a sum $x$ in Fig.\ref{fg1} is called  an $\Lambda$-type sum if $\Lambda={\rm supp}(x)$. For $\Lambda\subseteq [M]-\{\theta\}$, define $\underline{\Lambda}=\Lambda\cup\{\theta\}$ and call $\Lambda$-type sums as interference. Let $\gamma^{(i)}_k$ be the number of $\Lambda$-type sums in ${\rm Serv}^{(i)}$ for each $k$-subset $\Lambda\subseteq[3]$ and $i\in[5]$.

 \begin{figure}[ht]
\centering
\setlength{\abovecaptionskip}{0cm}
\setlength{\belowcaptionskip}{-0.2cm}
\begin{tikzpicture}[scale=1]
\draw(-4.5,2.6)--(-4.5,-2.55);
\draw(-2.85,2.6)--(-2.85,-2.55);
\draw(-1.25,2.6)--(-1.25,-2.55);
\draw(0.35,2.6)--(0.35,-2.55);
\draw(2.35,2.6)--(2.35,-2.55);
\draw(4.3,2.6)--(4.3,-2.55);
\draw[dotted](-4.45,1.65)--(4.25,1.65);
\draw[dotted](-4.45,0.95)--(4.25,0.95);
\draw[dotted](-4.45,-0.45)--(4.25,-0.45);
\draw[dotted](-4.45,-1.25)--(4.25,-1.25);
\node at (-0.1,0){\footnotesize\begin{tabular}{ccccc}
\specialrule{0.09em}{0pt}{1.5pt}
$\rm{Serv}^{(1)}$ & $\rm{Serv}^{(2)}$ & $\rm{Serv}^{(3)}$& $\rm{Serv}^{(4)}$& $\rm{Serv}^{(5)}$\\\specialrule{0.09em}{1pt}{1.5pt}
$a_1$&$a_4$&$a_7$& $a_{10},a_{11},a_{12}$&$a_{13},a_{14},a_{15}$\\ \specialrule{0em}{1.6pt}{1.6pt}
$b_1$&$b_4$&$b_7$& $b_{10},b_{11},b_{12}$&$b_{13},b_{14},b_{15}$\\
$c_1$&$c_4$&$c_7$& $c_{10},c_{11},c_{12}$&$c_{13},c_{14},c_{15}$\\ \specialrule{0em}{1.6pt}{1.6pt}
$a_{2}+b_{2}$&$a_{5}+b_{5}$&$a_{8}+b_{8}$&&\\
$a_{3}+b_{3}$&$a_{6}+b_{6}$&$a_{9}+b_{9}$&&\\
$a_{16}+c_{2}$&$a_{18}+c_{5}$&$a_{20}+c_{14}$&&\\
$a_{17}+c_{3}$&$a_{19}+c_{6}$&$a_{21}+c_{15}$&&\\\specialrule{0em}{1.6pt}{1.6pt}
$b_{16}+c_{16}$&$b_{18}+c_{18}$&$b_{20}+c_{20}$&&\\
$b_{17}+c_{17}$&$b_{19}+c_{19}$&$b_{21}+c_{21}$&&\\\specialrule{0em}{1.6pt}{1.6pt}
&&&\makecell{$a_{22}+b_{22}$\\$+c_{22}$\\$a_{23}+b_{23}$\\$+c_{23}$}&\makecell{$a_{24}+b_{24}$\\$+c_{24}$\\$a_{25}+b_{25}$\\$+c_{25}$}\\

\specialrule{0.09em}{1pt}{1.5pt}
\end{tabular}};
\end{tikzpicture}
\caption{Answers of the $(N=5,T=3,M=3)$ PIR scheme for retrieving $W_1$.}
\label{fg1}
\end{figure}
Now we show the scheme satisfies the correctness condition and the privacy condition.
Recall the sufficient conditions for the correctness (s1) requirement in \cite{Zhang&Xu17:OptimalSubpacketization}, that is, for any $\Lambda\subseteq [M]-\{\theta\}$, the interference parts of all $\underline{\Lambda}$-type sums can be linearly expressed by the $\Lambda$-type sums which appears in all servers. For $\Lambda\subseteq[M]-\{\theta\}$, we collect all $\Lambda$-type sums and the interference parts of all $\underline{\Lambda}$-type sums to form a matrix and call this matrix as the distribution matrix of $\Lambda$-type sums.  For example, for $\{2\}$-type sums $\{b_i\}$, its distribution matrix has the following form,
\begin{equation}\label{eq2}
\begin{pmatrix}{\bf b_1}&{\bf b_4}&{\bf b_7}&{\bf b_{10}}&{\bf b_{13}}\\
b_{2}&b_{5}&b_{8}&{\bf b_{11}}&{\bf b_{14}}\\b_{3}&b_{6}&b_{9}&{\bf b_{12}}&{\bf b_{15}}
\end{pmatrix}.
\end{equation}
where the bold symbols are all $\{2\}$-type sums and the rest are the interference parts of all $\{2,3\}$-type sums.
Similarly, the distribution matrix of $\{2,3\}$-type sums $\{b_i+c_j\}$ is
\begin{equation}\label{eq3}
{\footnotesize\begin{pmatrix}
{\bf b_{16}+c_{16}}&{\bf b_{18}+c_{18}}&{\bf b_{20}+c_{20}}&b_{22}+c_{22}&b_{24}+c_{24}\\
{\bf b_{17}+c_{17}}&{\bf b_{19}+c_{19}}&{\bf b_{21}+c_{21}}&b_{23}+c_{23}&b_{25}+c_{25}
\end{pmatrix}.}
\end{equation}
Then by the {\rm MDS} property of $G_1$ and $G_2$, the coordinates labeled by the bold symbols in (\ref{eq2}) and (\ref{eq3}) form an information set of $G_1$ and $G_2$, respectively. That is, the rest symbols can be recovered by the bold symbols in (\ref{eq2}) and (\ref{eq3}), respectively.  Note that for any $k$-subset $\Lambda\subseteq [M]-\{\theta\}$, the distribution matrix of $\Lambda$-type sums in Fig.\ref{fg1} has a similar form, as the matrix (\ref{eq2}) or (\ref{eq3}), so the interference parts of all $\underline{\Lambda}$-type sums  can be recovered by all $\Lambda$-type sums appeared in $5$ columns. Hence this scheme satisfies the condition (s1) in \cite{Zhang&Xu17:OptimalSubpacketization}, i.e, the correctness condition is guaranteed.

As to the privacy, recall the sufficient conditions for the privacy (s2) requirement in \cite{Zhang&Xu17:OptimalSubpacketization}, it is sufficient to ensure that for any $\Lambda\subseteq [M]-\{\theta\}$, there are the same number of independent symbols contained in any $3$ columns of $\Lambda$-type sums' distribution matrix (i.e., $(\ref{eq2})$ or $(\ref{eq3})$).  Actually this is guaranteed by the {\rm MDS} property of the linear code which is used to generate such type interference. Thus the privacy condition is guaranteed.

 Moreover the desired record consists of $25$ symbols while the answers totally contain $49$ symbols, so the scheme has rate $\frac{25}{49}$ attaining the
capacity for this case.
\end{example}
The field size relies on the maximum length of the {\rm MDS} codes used in this scheme, so it requires $q\geq 15$ in Example \ref{eg1}. Note that for any $\Lambda$-type interference, if its distribution matrix is a codeword of some {\rm MDS} array code, then there are also the same number of independent symbols contained in any $3$ columns of its distribution matrix. For example, suppose $\Lambda=\{2\}$,  the matrix (\ref{eq2}) can be viewed as a codeword of an $(5,3;3)$ {\rm MDS} array code. Similarly, the $[10,6]$ {\rm MDS} code also can be viewed as an $(5,3;2)$ {\rm MDS} array code. So if we adopt  $(5,3;3)$ and $(5,3;2)$ {\rm MDS} array codes rather than $[15,9]$ and $[10,6]$ {\rm MDS} codes, then the new obtained scheme also satisfies the $T$-privacy condition.

 However, there is a problem that how to guarantee the correctness condition. So the MDS array codes have to satisfy some property determined by the correctness condition. More precisely, for any $k$-subset $\Lambda$ of $[M]-\{\theta\}$ and {\rm MDS} array code corresponding to the $\Lambda$-type sums, denoted by $(5,3;\ell_k)$, its generator matrix needs to have the following recovery property:
 \begin{itemize}
  \item [(a1)]\emph{for $i$th thick column, there are $\gamma^{(i)}_k$ columns which are used to generate the $\Lambda$-type sums, and the rest $\gamma^{(i)}_{k+1}$ columns are used to generate  $\underline{\Lambda}$-type sums, that is, $\ell_k=\gamma^{(i)}_k+\gamma^{(i)}_{k+1}$.}
  \item [(a2)]\emph{ All these $\sum_{i=1}^N\gamma^{(i)}_k$ columns have full column rank, that is, $\sum_{i=1}^N\gamma^{(i)}_k=T\ell_k$. }
\end{itemize}
Now we give two admissible matrices $\tilde{G}_1\in\mathbb{F}_2^{9\times 15},\tilde{G}_2\in\mathbb{F}_2^{6\times 10}$.  That is,
{\footnotesize
\begin{align}
\tilde{G}_1&=\begin{pmatrix}
{0}&1&0&{0}&1&0&{1}&1&0&{1}&{0}&{0}&{1}&{0}&{0}\\
{1}&0&1&{1}&0&0&{1}&0&1&{0}&{1}&{0}&{0}&{1}&{0}\\
{1}&0&0&{0}&0&1&{1}&1&0&{0}&{0}&{1}&{0}&{0}&{1}\\
{0}&1&0&{0}&0&1&{0}&1&1&{1}&{1}&{1}&{1}&{1}&{0}\\
{1}&0&1&{0}&1&0&{1}&1&1&{0}&{1}&{1}&{1}&{1}&{1}\\
{1}&0&0&{1}&0&1&{1}&0&1&{1}&{1}&{0}&{1}&{0}&{1}\\
{0}&1&0&{1}&0&1&{0}&1&0&{0}&{1}&{0}&{0}&{0}&{1}\\
{1}&0&1&{0}&0&1&{1}&0&1&{1}&{0}&{1}&{1}&{0}&{0}\\
{1}&0&0&{1}&1&1&{0}&0&1&{1}&{0}&{0}&{0}&{1}&{1}
\end{pmatrix},\notag\\
\tilde{G}_2&=\begin{pmatrix}
{1}&{0}&{1}&{0}&{1}&{0}&1&0&0&0\\
{0}&{1}&{0}&{1}&{0}&{1}&0&1&0&0\\
{1}&{0}&{0}&{1}&{1}&{1}&0&0&0&0\\
{0}&{1}&{1}&{1}&{1}&{0}&0&0&0&0\\
{1}&{0}&{1}&{1}&{0}&{1}&0&0&1&0\\
{0}&{1}&{1}&{0}&{1}&{1}&0&0&0&1 \notag
\end{pmatrix}.
\end{align}}Then one can verify that the columns labeled by $\{1,4,7,10,11,12,13,14,15\}$ in $\tilde{G}_1$  have  full column rank and the columns labeled by $\{1,2,3,4,5,6\}$ in $\tilde{G}_2$  also have full column rank. Hence $\tilde{G}_1$ and $\tilde{G}_2$ satisfy the recovery property. Actually, $\tilde{G}_1$ is obtained by applying the method in Theorem \ref{thm0} to a generator matrix of an $[5,3]$ {\rm Generalized  Reed-Solomon} code  over $\mathbb{F}_{2^3}$ and rearranging the order of columns in each thick column by multiplying some $3\times 3$ permutation matrix. Similarly, $\tilde{G}_2$ is obtained by using the same method  to a generator matrix of a $[5,3]$ doubly-extended {\rm Generalized  Reed-Solomon} code over $\mathbb{F}_{2^2}$. Then the {\rm MDS} property of $\tilde{G}_1$ and $\tilde{G}_2$ is also satisfied.  Therefore, the new scheme obtained by using $\tilde{G}$ to replace $G$ in (\ref{eqb}) is a capacity-achieving $T$-PIR scheme with optimal sub-packetization, where {\small$\tilde{G}={\footnotesize\begin{pmatrix}
\tilde{G}_1&0\\0&\tilde{G}_2
\end{pmatrix}}$}.  Note that the field size is reduced to $2$.

As displayed in the example, the main design idea behind our scheme is to make each $(N,T;\ell_k)$ {\rm MDS} array code corresponding to $\Lambda$-type interference in the scheme satisfy the recovery property for any $k$-subset $\Lambda$ of $[M]-\{\theta\}$ and some fixed $\gamma^{(1)}_k,...,\gamma^{(N)}_k$. Fortunately, we prove that every {\rm MDS} array code trivially satisfies the recovery property by Lemma \ref{lem3} in Section \ref{sec4a}.

\section{ The General $T$-PIR Scheme Based On {\rm MDS} Array Codes}\label{sec4}
In this section we first characterize the recovery property of  MDS array codes  and  then describe our general capacity-achieving  $T$-PIR scheme based on {\rm MDS} array codes.

\subsection{The Recovery property of MDS array code}\label{sec4a}

\begin{lemma}\label{lem3} Suppose {\small$G=(G^{(1)},G^{(2)},...,G^{(N)})\in\mathbb{F}_q^{T\ell\times N\ell}$} is a generator matrix of an $(N,T;\ell)$ MDS array code over $\mathbb{F}_q$, where {\small$G^{(i)}=({\bf g}^{(i)}_{1},...,{\bf g}^{(i)}_{\ell})\in \mathbb{F}_q^{T\ell\times \ell}$} and ${\bf g}^{(i)}_{j}$ is a $T\ell$-length column vector. Then for any $(m_1,...,m_N)\in\{0,1,...,\ell\}^{N}$ with $\sum_{i=1}^Nm_i=T\ell$, there exist $N$ subsets $\Gamma_1,...,\Gamma_N$ of $[\ell]$ with $|\Gamma_i|=m_i$ such that
\begin{align*}
{\rm rank}(G_{\Gamma_1}^{(1)},G_{\Gamma_2}^{(2)},...,G_{\Gamma_N}^{(N)})=T\ell.
\end{align*}
\end{lemma}
\begin{proof} For any fixed $(m_1,...,m_N)$ with $\sum_{i=1}^Nm_i=T\ell$, there exist at least $T$ nonzero numbers of them. Without loss of generality, we may assume that $m_i\neq 0, 1\leq i\leq N$. Because that if $m_i=0$, the new matrix obtained by deleting the thick block column $G^{(i)}$ is also a generator matrix of an $(N-1,T;\ell)$ MDS array code.

Let
{\small\begin{align*} b(m_1,m_2,...,m_N)=\max_{\substack{\Gamma_1,...,\Gamma_N\subseteq[\ell],\\|\Gamma_i|=m_i,1\leq i \leq N.}}{\rm rank}(G_{\Gamma_1}^{(1)},G_{\Gamma_2}^{(2)},...,G_{\Gamma_N}^{(N)}).
\end{align*}}Then there exist $N$ subsets $\Gamma_1,...,\Gamma_N$ of $[\ell]$ with $|\Gamma_i|=m_i$ such that {\small$b={\rm rank}(G_{\Gamma_1}^{(1)},G_{\Gamma_2}^{(2)},...,G_{\Gamma_N}^{(N)})$}. Choose a maximum linearly independent subset of the vectors $\{G_{\Gamma_1}^{(1)},G_{\Gamma_2}^{(2)},...,G_{\Gamma_N}^{(N)}\}$, denoted by $\{G_{\Gamma^{(1)}_i}^{(i)}:i\in[N]\}$, then $b={\rm rank}(G_{\Gamma^{(1)}_1}^{(1)},G_{\Gamma^{(1)}_2}^{(2)},...,G_{\Gamma^{(1)}_N}^{(N)}),$ where $\Gamma^{(1)}_i\subseteq\Gamma_i$ for $i\in[N]$ and $\sum^N_{i=1}|\Gamma^{(1)}_i|=b$.  Then it is sufficient to show that $b=T\ell$.

On the contrary, we assume that $b<T\ell$. Let  ${\bf V}=\textsl{Cspan}(G_{\Gamma_1}^{(1)},G_{\Gamma_2}^{(2)},...,G_{\Gamma_N}^{(N)})$, where $\textsl{Cspan}(\cdot)$ denotes the linear space spanned by all columns of the matrix over $\mathbb{F}_q$. Then $\dim {\bf V}=b<T\ell$ and $\{G_{\Gamma^{(1)}_i}^{(i)}:i\in[N]\}$ is a base of the vector space ${\bf V}$. To derive a contradiction, we assume the following  claim has been proved.

{\bf Claim :}  \emph{if $b<T\ell$, then for $1\leq f\leq T$, there exist $f$ disjoint nonempty subsets $X_1,...,X_f$ of $[N]$ such that\\
 $\forall u\in \bigcup_{i=1}^fX_i, \forall j\in[\ell], {\bf g}^{(u)}_{j}\in {\bf V}.
 $}

 Particularly, let $f=T$.  Then it follows from the {\bf Claim} that there exist $T$ disjoint nonempty subsets
 $X_1,...,X_T$ of $[N]$ such that
 $\forall u\in \bigcup_{i=1}^TX_i, \forall j\in[\ell], {\bf g}^{(u)}_{j}\in {\bf V}.$
 Hence
  ${\rm rank}(G^{(u)}:u\in \bigcup_{i=1}^TX_i)\leq \dim {\bf V}=b.$
  On the other hand, note that $|\bigcup_{i=1}^TX_i|=\sum_{i=1}^T|X_i|\geq T$.  Combining with the MDS property of $G$, then ${\rm rank}(G^{(u)}:u\in \bigcup_{i=1}^TX_i)=T\ell$. So one can obtain that $T\ell\leq b< T\ell$, a contradiction.

   To complete the proof, it remains to prove the {\bf Claim}. Now we prove it by induction on $f$.

  For $f=1$, let $X_1=\{i\in[N]: \Gamma^{(1)}_i\neq \Gamma_i\}$. Since $b<T\ell$, then $|X_1|\geq 1$. For any $ u\in X_1$, it is sufficient to show that for $j\in [\ell]-\Gamma_u$, ${\bf g}^{(u)}_{j}\in V$. Then choosing a $m_u$-subset $\Gamma_u^\prime$ of $[\ell]$ such that $\Gamma_u^{(1)}\cup\{j\}\subseteq \Gamma_u^\prime$, one can obtain that
  {\small $${\bf V}\subseteq \textsl{Cspan}(G_{\Gamma_1}^{(1)},...,G_{\Gamma_{u-1}}^{(u-1)},G^{(u)}_{\Gamma_u^\prime},G_{\Gamma_{u+1}}^{(u+1)}...,G_{\Gamma_N}^{(N)}). $$}By the definition of $b(m_1,m_2,...,m_N)$, it holds that
  \begin{align*}{\rm rank}(G_{\Gamma_1}^{(1)},...,G_{\Gamma_{u-1}}^{(u-1)},G^{(u)}_{\Gamma_u^\prime},G_{\Gamma_{u+1}}^{(u+1)}...,G_{\Gamma_N}^{(N)})\leq b=\dim{\bf V},
  \end{align*}which implies that
  { $${\bf V}= \textsl{Cspan}(G_{\Gamma_1}^{(1)},...,G_{\Gamma_{u-1}}^{(u-1)},G^{(u)}_{\Gamma_u^\prime},G_{\Gamma_{u+1}}^{(u+1)}...,G_{\Gamma_N}^{(N)}). $$}Hence, ${\bf g}^{(u)}_{j}\in {\bf V}$.

 Suppose that there exist $f-1$ disjoint nonempty subsets $X_1,...,X_{f-1}$ of $[N]$ such that\vspace{-0.2cm}
 \begin{align*}
 \forall u\in \bigcup_{i=1}^{f-1}X_i, \forall j\in[\ell], {\bf g}^{(u)}_{j}\in {\bf V}.
 \end{align*}
  Consider the case $f$, note that
  ${\rm rank}(G^{(u)}:u\in\bigcup_{i=1}^{f-1}X_i)\leq \dim{\bf V}=b<T\ell.$
  By the {\rm MDS} property (\ref{eqmds}) of $G$,  then the  $|\bigcup_{i=1}^{f-1}X_i|<T$, which implies that the vectors $\{{\bf g}^{(u)}_{j}:u\in \bigcup_{i=1}^{f-1}X_i, j\in[\ell]\}$ are linearly independent over $\mathbb{F}_q$, so are the vectors $\{{\bf g}^{(u)}_{j}:u\in \bigcup_{i=1}^{f-1}X_i, j\in\Gamma_u\}$. Therefore the vectors can extend to be a base of $\bf V$. Then there exist $\Gamma_i^{(f)}\subseteq\Gamma_i$ for $1\leq i\leq N$ and $\Gamma_u^{(f)}=\Gamma_u$ for $u\in \bigcup_{i=1}^{f-1}X_i$
  such that {\small$\{G_{\Gamma^{(f)}_i}^{(i)}:i\in[N]\}$} is a base of the vector space ${\bf V}$.  Let {\small$X_f=\{i\in[N]: \Gamma^{(f)}_i\neq \Gamma_i\}$}. Then $X_f\neq\emptyset$, otherwise {\small$\dim{\bf V}=\sum^N_{i=1}|\Gamma_i|=T\ell$}. By the definition of $X_f$, one can obtain that {\small$X_f \bigcap (\bigcup_{i=1}^{f-1}X_i)=\emptyset$}, that is, such $f$ subsets $X_i,i\in[f]$ are disjoint. Similarly, by using the same way in the case $f=1$, one can obtain that {\small$\forall u\in X_f,\forall j\in[\ell],{\bf g}^{(u)}_{j}\in {\bf V}$.}
\end{proof}
\begin{remark}\label{rmk2} Using the notations introduced above, we  may assume that for any fixed $(m_1,m_2,...,m_N)$ with $\sum^N_{i=1}m_i=T\ell$,
$\Gamma_i\!=\!\{1,2,...,m_i\},i\!\in\![N]$ in a generator matrix of the $(N,T;\ell)$ MDS array code. This is because that we can rearrange the order of $\ell$ columns in each thick column by multiplying some $\ell\times\ell$ permutation matrix.
\end{remark}

\subsection{Formal Description of the general scheme}

Our scheme can be obtained by modifying the capacity-achieving  $T$-PIR schemes in \cite{Zhang&Xu17:OptimalSubpacketization}.  As in  Example \ref{eg1}, we replace  $M-1$ {\rm MDS} codes with some $M-1$ {\rm MDS} array codes. Next we give these $M-1$ desired {\rm MDS} array codes.

 Specially, for $1\leq k\leq M-1$, the $k$th {\rm MDS} code defined by the generator matrix $G_k$ in \cite{Zhang&Xu17:OptimalSubpacketization} has the parameters $[\frac{N}{T}(T\alpha_k+(N-T)\beta_k),T\alpha_k+(N-T)\beta_k]$ over $\mathbb{F}_q$, where $\alpha_k,\beta_k$ are defined as in the identities $(35),(36)$ in \cite{Zhang&Xu17:OptimalSubpacketization}. Note that
$T\alpha_k+(N-T)\beta_k=T(n-t)^{k-1}t^{M-1-k}$, and  define $\ell_k=(n-t)^{k-1}t^{M-1-k}$,
 where $t=T/{\rm gcd}(N,T),~n=N/{\rm gcd}(N,T)$. Then for $1\leq k\leq M-1$, the $k$th {\rm MDS} code can be viewed as an $(N,T;\ell_k)$ {\rm MDS} array code.  By Lemma \ref{lem3},  one can choose a generator matrix $\tilde{G}_k$ of an $(N,T;\ell_k)$ {\rm MDS} array code which has the recovery property for {\small$(\gamma^{(1)}_k,...,\gamma^{(N)}_k)$,} where $\gamma^{(i)}_k=\alpha_k$ for $1\leq i \leq T$ and $\gamma^{(i)}_k=\beta_k$ for $T+1\leq i \leq N$. Then these $M-1$ matrices $\tilde{G}_k$ are desired.

 One can verify that the new scheme satisfies the correctness condition and $T$-privacy condition, which are guaranteed by the recovery property and MDS property of all $M-1$  {\rm MDS} array codes, respectively.  Moreover, the new scheme doesn't change the sub-packetization of records and download size. Therefore the new scheme has the highest rate and the optimal sub-packetization. Note that there are $M-1$ {\rm MDS} array codes used in our scheme over $\mathbb{F}_q$, by Theorem \ref{thm0} it only needs to requires that for $1\leq k\leq M-1$, $q^{\ell_k}\geq N$. That is, $q\geq \sqrt[\ell]{N}$, where $\ell=\min\{t^{M-2},(n-t)^{M-2}\}$.

\section{Conclusion}\label{sec5}
In this paper we build a general capacity-achieving $T$-PIR scheme  based on  {\rm MDS} array codes over $\mathbb{F}_q$, that is, the queries are generated by using $M-1$ {\rm MDS} array codes rather than {\rm MDS} codes. It requires the field size  $q\geq \sqrt[\ell]{N}$ and has optimal sub-packetization.  In particular, the binary field  is enough to build our scheme as long as $M\geq 2+\lceil\log_\mu\log_2N\rceil$, where $\mu=\min\{t,n-t\}>1$.

\end{document}